\documentclass[11pt,a4paper,reqno]{amsart}%
\usepackage{amsthm,amsmath,amsfonts,amssymb,amsxtra,bookmark,latexsym,hyperref,delarray,a4,color,euscript}

\theoremstyle{plain}
\newtheorem{theorem}{Theorem}[section]
\newtheorem{lemma}[theorem]{Lemma}

\theoremstyle{definition}

\theoremstyle{remark}

\numberwithin{equation}{section}

\title[Asymptotic formula  of the Stokes problem]{Asymptotic formula for the solution of the Stokes problem with a small perturbation of the domain in two and three dimensions}

\author[T.H.C. LUONG]{Thi Hong Cam LUONG}
\address{CNRS \& Laboratoire de Math\'ematiques (UMR 8088), Universit\'e de Cergy-Pontoise, F-95000 Cergy-Pontoise, France.}
\email{thi-hong-cam.luong@u-cergy.fr}

\author[C. DAVEAU]{Christian DAVEAU}
\address{CNRS \& Laboratoire de Math\'ematiques (UMR 8088), Universit\'e de Cergy-Pontoise, F-95000 Cergy-Pontoise, France.}
\email{christian.daveau@u-cergy.fr}

\date{April 8, 2013}


\begin{document}

\begin{abstract}
In this paper we consider  the resolvent Stokes problem in the case there is a small perturbation of the domain caused by a perturbed boundary. Firstly, we prove that the solution of Stokes problem is continuous due to this small perturbation. Secondly, we derive the first-order term in the displacement field perturbation that due to the deformation of the domain. It is worth emphasizing that
even though only the first-order term is given, 
our method enables us
to derive higher-order terms as well. The derivation is rigorous and based on layer potential  techniques.\end{abstract}

\maketitle

\setcounter{tocdepth}{2}
\tableofcontents

\section{Introduction}

Let $\Omega \subset {\mathbb{R}}^d~( d=2,3)$  be a bounded domain with
boundary $\partial \Omega$ of class ${\mathcal {C}}^2$. We consider the  Stokes resolvent system
\begin{equation}\label{1}
\left \{
\begin{array}{lll}
-\Delta u + \nabla p = \lambda u & \mbox{in} & \Omega\\
\nabla \cdot u=0 & \mbox{in} & \Omega\\
u=g & \mbox{on} & \partial \Omega.
\end{array} \right.
\end{equation}
with $u(x)$ is the velocity of the fluid, $p(x)$ is the pressure, and $g(x)$ is the boundary condition.

Denote by $\nu$ the outward unit normal to $\partial \Omega$, let $\partial{\Omega_\delta}$ be a $\delta-$perturabation of $\partial\Omega$ as follows:\\
 In the  two-dimensional case, let $\rho\in {\mathcal {C}}^1(\partial \Omega)$, we consider the perturbed boundary $\partial \Omega_\delta$ be given by
\[ \partial \Omega_\delta = \{\widetilde{x}:\widetilde{x}=x+\delta \rho(x)\nu (x), x \in \partial \Omega \};\]
 and in the three-dimensional case, we consider  the parallel surface
\[ \partial \Omega_\delta = \{\widetilde{x}:\widetilde{x}=x+\delta \nu (x), x \in \partial \Omega\}.\]

Here we consider the parameter $\delta > 0$ which tends to $0$.
We assume that $g$ is an analytic function in a small neighborhood of $\partial \Omega$, and denote by $g_\delta$  the extended function of $g$ onto $\partial \Omega_\delta$.
Then denote by  $u_\delta$  the solution of the perturbation problem
\begin{equation}\label{new}
\left \{
\begin{array}{lll}
-\Delta u_\delta + \nabla p_\delta = \lambda u_\delta & \mbox{in} & \Omega_\delta \\
\nabla \cdot u_\delta=0 & \mbox{in} & \Omega_\delta\\
u_\delta =g_\delta & \mbox{on} & \partial {\Omega_\delta}.
\end{array} \right.
\end{equation}

In this paper we  verify the continuity of the solution $u_\delta$ with respect to $\delta$  and derive the asymptotic expansion of $(u_\delta-u)|_{\Omega^\circ}$ as $\delta$ tends to $0$, where $\Omega^\circ$ is any closed subset of $\Omega\cap \Omega_\delta$.
The asymptotic expansion is derived due to the theory of layer potentials and Fredholm's alternative, and the properties of
small perturbation of an interface. In connection with this, we refer to recent works in the context
of interface problems [1], [2], [7],  the continuity of the solution due to a small perturbation of an interface [7] and layer-potential theory for Stokes problem [3], [4], [5], [6], [8], [9]. In the work of H.~Ammari, H.~Kang, M. ~Lim, and H.~Zribi [1] , they derived the asymptotic expansion of the boundary perturbations of steady-state voltage potentials, and now in our work, we derive the asymptotic expansion for the Stokes problem with Dirichlet boundary condition. However, by the same method, we can derive asymptotic formula for
the Neumann problem as well.

This paper is organized as follows. In the next section, we introduce some notations for
small perturbations of an interface of class ${\mathcal{C}}^2$, review some basic facts on the layer-potentials
and give representation formulas. In Section 3, we verify the continuity of the solution
with respect to $\delta$, and in Section 4, we derive asymptotic expansion for the displacement field perturbation in
term of $\delta$.
\section{Definitions and Preliminary results}
\subsection{Small perturbation of an interface}
\subsubsection{In the case $d=2$}
 Let $a, b \in \mathbb{R}$, with $a<b$, and let
$X(t):[a,b]\rightarrow \mathbb{R}^2$ be the arclength
parametrization of $\partial \Omega$, namely, $X$ is a
${\mathcal{C}}^2$-function satisfying $|X'(t)=1|$ for all $t\in [a,b]$
and
$$\partial \Omega :=\{x=X(t), t\in [a,b]\}.$$
Then the outward unit normal to $\partial \Omega$, $\nu(x)$, is given by $\nu(x)=R_{\frac{-\pi}{2}}X'(t)$, where $R_{\frac{-\pi}{2}}$ is the rotation by $-\pi/2$; the tangential vector at $x$, $T(x)=X'(t)$, and $X'(t)\bot X''(t)$. Set the curvature $\tau(x)$ to be defined by
$$X''(t)=\tau(x)\nu(x).$$ 

We will sometimes use $\rho(x)$ for $\rho(X(t))$ and $\rho'(x)$ for the tangential derivative of $\rho(x)$. \\

Then, $\widetilde{x}=\widetilde{X}(t)=X(t)+\delta \rho(x) \nu (x)=X(t)+\delta \rho(x)R_{\frac{-\pi}{2}}X'(t) $ is a parametrization of $\partial \Omega_\delta$.  We denote by ${\widetilde \nu}(\widetilde{x})$ the outward unit normal to $\partial  \Omega_\delta$ at $\widetilde{x}$. Then, it is proved in $[1]$ that $\widetilde{\nu}(\widetilde{x})$ can be expanded uniformly as
\begin{equation}\label{nv}
\widetilde{\nu}(\widetilde{x})=\sum^\infty_{n=0}\delta^n\nu^n(x),~\mbox{with} ~\widetilde{x}=x+\delta \rho(x)\nu (x),~\widetilde{x} \in \partial \Omega_\delta,~x \in \partial \Omega,
\end{equation}
where the vector-valued functions $\nu^n(x)$ are uniformly bounded regardless of $n$. In particular,
$$\nu^0(x)=\nu(x), \nu^1(x)=-\rho(x)T(x), \hspace{5mm} x\in \partial \Omega.$$

Likewise, denote by $d\sigma_\delta (\widetilde{x})$  the length element of $\partial \Omega_\delta$ at $\widetilde{x}$, which has an uniform expansion (see in [1])
\begin{equation}\label{length}
d\sigma_\delta (\widetilde{x})=|\widetilde{X}'(t)|dt=\sqrt{(1-\delta \tau (t)\rho(t))^2+\delta^2\rho'^2(t)}dt=\sum_{n=0}^\infty\delta^n\sigma^n(x)d\sigma(x), 
\end{equation}
 with $\widetilde{x} \in \partial \Omega_\delta,~x\in \partial \Omega$, and $\sigma^n$ are functions bounded regardless of $n$, with
 \[\sigma^0(x)=1, \sigma^1(x)=-\tau(x)\rho(x), x\in \partial \Omega.\]

 Let $\widetilde{x},\widetilde{y}\in \partial \Omega_\delta$, that is $\widetilde{x}=x+\delta\rho(x)\nu(x),~\widetilde{y}=y+\delta \rho(y)\nu(y)$, then
  $$\widetilde{x}-\widetilde{y}=x-y + \delta (\rho(x)\nu(x)-\rho(y)\nu(y)),$$
  and
  $$|\widetilde{x}-\widetilde{y}|^2=|x-y|^2(1+2\delta\frac{\langle x-y, \rho(x)\nu(x)-\rho(y) \nu(y)\rangle}{|x-y|^2} +\delta^2 \frac{|\rho(x)\nu(x)-\rho(y) \nu(y)|^2}{|x-y|^2}).$$
  Denote by
  $$E(x,y):=\frac{\langle x-y, \rho(x)\nu(x)-\rho(y) \nu(y)\rangle}{|x-y|^2}, G(x,y):=\frac{|\rho(x)\nu(x)-\rho(y) \nu(y)|^2}{|x-y|^2}.$$
Since $\partial \Omega$ is of class ${\mathcal{C}}^2$, we can see that
\begin{equation}\label{rm2}
 |E(x,y)|+|G(x,y)|^{\frac{1}{2}}\leq C\|X\|_{{\mathcal{C}}^2(\partial\Omega)}\|\rho\|_{{\mathcal{C}}^1(\partial \Omega)},~for~all~x,y\in\partial \Omega,
\end{equation}
and hence
$$|\widetilde{x}-\widetilde{y}|=|x-y|\sqrt{1+2\delta E(x,y)+\delta^2 G(x,y)}=|x-y|\sum_{n=0}^\infty\delta^nL_n(x,y),$$
where the series converges absolutely and uniformly. In particular, we can see that
$$L_0(x,y)=1, L_1(x,y)=E(x,y).$$

\subsubsection{In the case $d=3$}
Denote by $d\sigma_\delta(\widetilde{x})$ the surface element  of $\partial \Omega_\delta$  at $\widetilde{x}$. Thanks to the results in [2]  we get  the following  uniform expansion for $d\sigma_\delta(\widetilde{x})$
\begin{equation}\label{surface}
d\sigma_\delta(\widetilde{x})=\sum_{n=0}^\infty\delta^n\sigma^n(x)d\sigma(x), ~\mbox{with} ~\widetilde{x}=x+\delta \nu (x), \hspace{5mm}\widetilde{x} \in \partial \Omega_\delta,~x\in \partial \Omega.
\end{equation}
where $\sigma^n$ are functions bounded regardless of $n$, with
 \[\sigma^0(x)=1, \sigma^1(x)=-2H(x), \sigma^2(x)=  K (x),~ x\in \partial \Omega,\]
 with $H$ and $K$ denote the mean and Gaussian curvature of $\partial \Omega$ respectively. \\
In the other hand, due to the parallel property of $\partial\Omega$ and $\partial \Omega_\delta$ we have
\begin{equation}\label{nv3}
\widetilde{\nu}(\widetilde{x})=\nu(x),  ~\mbox{with} ~\widetilde{x}=x+\delta \nu (x),~\widetilde{x} \in \partial \Omega_\delta,~ x\in \partial \Omega.
\end{equation}
 Let $\widetilde{x},\widetilde{y}\in \partial \Omega_\delta$, then $\widetilde{x}-\widetilde{y}=x-y + \delta (\nu(x)-\nu(y))$, and
  $$|\widetilde{x}-\widetilde{y}|^2=|x-y|^2(1+2\delta\frac{\langle x-y, \nu(x)- \nu(y)\rangle}{|x-y|^2} +\delta^2 \frac{|\nu(x)- \nu(y)|^2}{|x-y|^2}).$$
  Denote by
  $$E(x,y):=\frac{\langle x-y, \nu(x)- \nu(y)\rangle}{|x-y|^2}, G(x,y):=\frac{|\nu(x)- \nu(y)|^2}{|x-y|^2}.$$
Since $\partial \Omega$ is of class ${\mathcal {C}}^2$, we can see that there exists a constant $C$ depending only on $\partial \Omega$ such that
\begin{equation}\label{rm3}
|E(x,y)|+|G(x,y)|^{\frac{1}{2}}\leq C,~{\textrm{for all}}~ x,y\in\partial \Omega,
\end{equation}
and hence
$$|\widetilde{x}-\widetilde{y}|=|x-y|\sqrt{1+2\delta E(x,y)+\delta^2 G(x,y)}=|x-y|\sum_{n=0}^\infty\delta^nL_n(x,y),$$
where the series converges absolutely and uniformly. In particular, we can see that
$$L_0(x,y)=1, L_1(x,y)=E(x,y).$$

\subsection{The potential theory for Stokes resolvent system}

We start to review some basic facts in the theory of layer potentials. We consider here $\lambda \in {\mathbb{C}} \setminus \{z \in {\mathbb{C}}: \mbox{Re} z \leq 0, \mbox{Im}z=0\}$. The fundamental tensors $(\Gamma,F)$ of the Stokes resolvent system~(\ref{1}) can be obtained by the Fourier transform method in the following forms (see [3] for $d=2$ and [4],~[8] for $d=3$):\\
In two dimensions, we have
  \begin{equation}\label{2}
\left \{
\begin{array}{ll}
\Gamma_{ij}(x)&= -\frac{1}{2\pi}\left\{\delta_{ij}e_1(\sqrt{\lambda} |x|)+\frac{x_ix_j}{|x|^2}e_2(\sqrt{\lambda} |x|)\right\},\\
F_i(x)&=-\frac{x_i}{2 \pi |x|^2},
\end{array} \right.
\end{equation}
for  $\forall x \in {\mathbb{R}}^2,~ x \neq 0$ with
$$~~~~~~~~~~~
e_1(\kappa)=K_0(\kappa)+\kappa^{-1}K_1(\kappa)-\kappa^{-2}$$
$$~~~~~~~~~~~~~e_2(\kappa)=-K_0(\kappa)-2\kappa^{-1}K_1(\kappa)+2\kappa^{-2},$$
where  $K_n~(n\in \mathbb{N}_0)$  denotes the modified Bessel function of order $n$, and $\delta_{ij}$ is the Kronecker' symbol.            \\

In three dimensions, we have
 \begin{equation}\label{2}
\left \{
\begin{array}{ll}
\Gamma_{ij}(x)&=-\frac{1}{4\pi}\{\delta_{ij}e_1(-\sqrt{\lambda}|x|)+\frac{x_ix_j}{|x|^3}e_2(-\sqrt{\lambda}|x|)\}, \\
F_i(x)&=-\frac{1}{4\pi}\frac{x_i}{|x|^3},
\end{array} \right.
\end{equation}


for $\forall x \in {\mathbb{R}}^3, x \neq 0$ with
$$e_1(\epsilon)=\sum^\infty_{n=0}\frac{(n+1)^2}{(n+2)!}\epsilon^n=exp(\epsilon)(1-\epsilon^{-1}+\epsilon^{-2})-\epsilon^{-2}$$
$$e_2(\epsilon)=\sum^\infty_{n=0}\frac{1-n^2}{(n+2)!}\epsilon^n=exp(\epsilon)(-1+3\epsilon^{-1}-3\epsilon^{-2})+3\epsilon^{-2}.$$

By using the notations ${\hat{x}}=x-y=(\hat x_1,\cdots,\hat x_d)$ and $r=|{\hat{x}}|$, we introduce the stress tensor ${\mathcal{S}}$  associated to the fundamental tensors $(\Gamma, F)$ and having the following components (see in [8]):
\begin{equation}\label{stresstensor}
S_{ijk}(x,y):=-F_j(\hat x)\delta_{ik} + \frac{\partial \Gamma_{ij}(\hat x)}{\partial \hat x_k}+\frac{\partial \Gamma_{kj}(\hat x)}{\partial \hat x_i},~~i,j,k=1,\cdots,d.
\end{equation}
Note that we used the Einstein convention for the summation notation omitting the summation sign for the indices appearing twice. We will continue using this convention throughout this paper. Taking into account the relation~(\ref{2}) and~(\ref{stresstensor})  we obtain the following explicit forms (see in [8]) as follows. \\
In two dimensions, for $\forall x,y \in {\mathbb{R}}^2,~ x \neq y$ we have:
\begin{equation}\label{f2d}
 S_{ijk}(x,y)=-\frac{1}{2 \pi}\left \{\delta_{ik} \frac{ \hat x_j}{r^2}d_1(\sqrt{\lambda}r)+(\delta_{kj}\frac{\hat x_i}{|x|^2}-\delta_{ij}\frac{\hat x_k}{r^2}) d_2(\sqrt{\lambda}r) \right\}
\end{equation}
$$\hspace{4cm}-\frac{1}{2 \pi}\left \{\frac{\hat x_i \hat x_j\hat x_k}{r^4}(2d_1(\sqrt{\lambda}r)+2d_2(\sqrt{\lambda}r)-2)\right \},$$
with
 $$d_1(\kappa)=2K_2(\kappa)+1-4\kappa^{-2}, d_2(\kappa)=2K_2(\kappa)+\kappa K_1(\kappa)-4\kappa^{-2}.$$
 In three dimensions, for $\forall x,y \in {\mathbb{R}}^3, x \neq y$, we have:
 \begin{equation}\label{f3d}
   S_{ijk}(x,y)= -\frac{1}{4\pi}\left \{\delta_{ik} \frac{\hat x_j}{r^3}d_1(-\sqrt \lambda r) - (\delta_{ij}\frac{\hat x_k}{r^3}+  \delta_{kj}\frac{\hat x_i}{r^3})d_2(-\sqrt \lambda r) \right \} 
\end{equation}
$$\hspace{4cm} -\frac{1}{4\pi} \left \{ \frac{\hat x_i\hat x_j\hat x_k}{r^5}(3-3d_1(-\sqrt \lambda r)+2d_2(-\sqrt \lambda r)) \right \},
$$
with
 $$d_1(\epsilon)=\sum^\infty_{n=2}\frac{2(n^2-1)}{(n+2)!}\epsilon ^n= exp(\epsilon)(2-6 \epsilon ^{-1}+6\epsilon ^{-2})-6\epsilon ^{-2} +1,$$
 $$d_2(\epsilon)=\sum^\infty_{n=2}\frac{n(n^2-1)}{(n+2)!}\epsilon ^n= exp(\epsilon)(\epsilon-3+6\epsilon ^{-1}-6\epsilon ^{-2})+6\epsilon ^{-2}.$$
 
The pressure tensor $\Lambda$ associated to the stress tensor ${\mathcal{S}}$ has the following components (see in [3] for $d=2$ and [4], [5] for $d=3$):
 \begin{equation}
 \Lambda_{ik}(x,y)=\left\{
 \begin{array}{ll}
 -\frac{1}{2\pi} (\delta_{ik}\lambda \ln r-4\frac{\delta_{ik}}{r^2}+8\frac{\hat x_i \hat x_k}{r^4}) & \textrm{for $d=2$} \\
 - \frac{1}{4\pi}(\frac{\delta_{ik}}{r^3}(\lambda r^2 -2)+ \frac{\hat x_i \hat x_k}{r^5})& \textrm{for $d=3$},
 \end{array}\right.
 \end{equation}
for $\forall x,y \in {\mathbb{R}}^d,~ x \neq y$.\\
Next we consider the single- and double-layer potentials associcated with stress and pressure tensors. Let $\partial D$ be a closed surface of class $C^2$ of a bounded  domain $D$ and $\phi=(\phi_1,\cdots,\phi_d)$ a vectorial continuous function on $\partial D$. For $x\in {\mathbb{R}}^d\backslash \partial D$, we define the single-layer potential ${\mathcal{V}}_{D}\phi$ which has components  as follows:
\begin{equation}\label{potential}
{\mathcal{V}}_{D}^{i}\phi (x):=  \int_{\partial D}  \Gamma_{ij}(x-y)\phi_j(y) d\sigma(y),~i=1,\cdots,d;
\end{equation}
and the double-layer potential ${\mathcal{W}}_ D\phi$ has components as follows
\begin{equation}\label{17}
{\mathcal{W}}_ D ^i \phi (x):=  \int_{\partial D} -S_{ijk}(x,y)\nu_k(y)\phi_j(y)d\sigma(y),~i=1,\cdots,d;
\end{equation}
where $\nu(y)$ is the outward unit normal vector to $\partial D$ at the point $y$. \\

Additionally, we consider the functions ${\mathcal{Q}}_D\phi$ and $\Pi_D \phi$ as follows
\begin{equation}\label{18}
 \begin{array}{ll}
{\mathcal{Q}}_D \phi (x):= & \int_{\partial D} F_i(x-y)\phi_i(y) d\sigma(y), \\
\Pi_D \phi (x):= & \int_{\partial D} \Lambda_{ik}(x-y)\nu_k(y)\phi_i(y) d\sigma(y)
\end{array}
\end{equation}
for $x \in {\mathbb{R}}^d\backslash \partial D$. \\

The functions  $({\mathcal{V}}_D \phi, {\mathcal{Q}}_D \phi)$ and $({\mathcal{W}}_D \phi, \Pi_D  \phi)$ are smooth functions in each of the domains ${\mathbb{R}}^d\backslash {\bar{D}_0}$ and $D_0$, respectively, where $D_0$ is the inner domain with the boundary $\partial D$. All these functions   satisfy the following equations of Stokes resolvent  problem: \\
\begin{equation}
\begin{array}{ll}
-\Delta {\mathcal{V}}_{D} \phi (x)+\nabla {\mathcal{Q}}_D \phi(x)=0, ~~\nabla\cdot{\mathcal{V}}_{D} \phi (x)=0,\\
-\Delta {\mathcal{W}}_{D} \phi (x)+\nabla \Pi_D \phi (x)=0, ~~\nabla\cdot{\Pi}_{D} \phi (x)=0,
\end{array}
\end{equation}
for $x\in {\mathbb{R}}^3\backslash \partial D$. \\

For our prupose, we also need to introcduce the principal value of the double-layer potential in a point $x_0$ of $\partial D$ defined by the following formula
\begin{equation}\label{k}
{\mathcal{K}}^i_D \phi (x_0):= p.v.\int_{\partial D}-S_{ijk}(x_0,y)\nu_k(y) \phi_j(y) d\sigma(y)
\end{equation}

For $\phi \in (C(\partial D))^d$, the following trace relation for ${\mathcal{W}}_D$ holds (see in [6]):
\begin{equation}
{\mathcal{W}}_D \phi \mid_\pm = (\mp \frac{1}{2}I+{\mathcal{K}}_D) \phi  \hspace{5mm} a.e.~ on~ \partial \Omega,
\end{equation}
where ${\mathcal{W}}_D \phi \mid_- $ and ${\mathcal{W}}_D \phi \mid_+ $ denote the limits from inside $D$ and outside $D$.\\

From the previous section, the solution for the interior  Dirichlet problem can be presented by the pure double-layer potential (see in [4]), 
then taking $D=\Omega$ and $D=\Omega_\delta$ in the fomulae~(\ref{potential}),~(\ref{17}),~(\ref{18}) and~(\ref{k}), we get the boundary integral representation of the solutions for the initial Stokes problem~(\ref{1}) and the perturbation problem~(\ref{new})  
are $({\mathcal{W}}_\Omega \phi, \Pi_\Omega \phi)$  and $({\mathcal{W}}_{\Omega_\delta}{\widetilde{\phi}}, \Pi_{\Omega_\delta}{\widetilde{\phi}})$, respectively.


Additionally, from~(\ref{f2d}) and~(\ref{f3d}), we conclude that the components $S_{ijk}(x,y)$ $            \nu_k(y)$ of ${\mathcal{S}}(x,y)\nu(y)$ can be written in the following manner (see in [3],[4]). \\

In two dimensions, for $\forall x,y \in {\mathbb{R}}^2, x \neq y$ , we have
\begin{equation}\label{Dij2}
{\mathcal{S}}_{ijk}(x,y)\nu_k(y)=-\frac{1}{2 \pi}\left \{ \frac{r_i\nu_j}{|r|^2}d_1(\sqrt{\lambda}|r|)+(\frac{\nu_ir_j}{|r|^2}-\delta_{ij}\frac{r\cdot \nu}{|r|^2}) d_2(\sqrt{\lambda}|r|)\right\},
\end{equation}
$$\hspace{4cm}-\frac{1}{2 \pi}\left \{ \frac{r_ir_jr\cdot \nu}{|r|^4}(2d_1(\sqrt{\lambda}|r|)+2d_2(\sqrt{\lambda}|r|)-2)\right\}$$
with
 $$d_1(\kappa)=2K_2(\kappa)+1-4\kappa^{-2}, d_2(\kappa)=2K_2(\kappa)+\kappa K_1(\kappa)-4\kappa^{-2}.$$

And in three dimensions, for $\forall x,y \in {\mathbb{R}}^3, x \neq y$, we have 
\begin{equation}
{\mathcal{S}}_{ijk}(x,y)\nu_k(y)=-\frac{1}{4\pi}\left \{ \frac{r_i\nu_j}{|r|^3}d_1(-\sqrt \lambda |r|) - (\frac{\nu_ir_j}{|r|^3}+  \delta_{ij}\frac{r\cdot \nu}{|r|^3})d_2(-\sqrt \lambda |r|) \right \} 
\end{equation}
$$\hspace{4cm}-\frac{1}{4\pi} \left \{ \frac{r_ir_jr \cdot \nu}{|r|^5}(3-3d_1(-\sqrt \lambda |r|)+2d_2(-\sqrt \lambda |r|)) \right \},
$$
with
 $$d_1(\epsilon)=\sum^\infty_{n=2}\frac{2(n^2-1)}{(n+2)!}\epsilon ^n= exp(\epsilon)(2-6 \epsilon ^{-1}+6\epsilon ^{-2})-6\epsilon ^{-2} +1,$$
 $$d_2(\epsilon)=\sum^\infty_{n=2}\frac{n(n^2-1)}{(n+2)!}\epsilon ^n= exp(\epsilon)(\epsilon-3+6\epsilon ^{-1}-6\epsilon ^{-2})+6\epsilon ^{-2}.$$
 
 For further purpose, we denote the kernel $-S_{ijk}(x,y)\nu_k(y)$ of double-layer potential ${\mathcal{K}}_D \phi$ by ${\mathcal{D}}_{ij}(x,y)$. We note that ${\mathcal{D}}_{ij}(x,y)$ behaves like ${\mathcal{O}}(r^{-d+1})$ as $r \rightarrow 0$ , which implies this kernel is weakly singular, so the integral operator is bounded on $({\mathcal{C}}(\partial D))^d$ (see in  [11, p. 243] or in [3], [4], [6]).

\section{Continuity of the solution with respect to the perturbation}

In this section we will show that the solution $u_\delta$ of the perturbation problem (\ref{new}) is continuous with respect to $\delta$ as $\delta$ tends to $0$. And for the guarantee of the solvability and uniqueness of the solutions for (\ref{1}) and~(\ref{new}), we assume the following compatibility conditions
$$\int_{\partial \Omega} g\cdot n d\sigma= 0,~\mbox{and} \int_{\partial \Omega_\delta} g_\delta \cdot n_\delta d\sigma_\delta= 0.$$
From the previous section, the  solution to the Stokes system~(\ref{1}) can be represented by the pure double-layer potential on $\partial \Omega$
\begin{equation}
 u^i(x)={\mathcal{W}}_ \Omega ^i \phi (x) =  \int_{\partial \Omega} -S_{ijk}(x,y)\nu_k(y)\phi_i(y)d\sigma(y),~x\in \Omega
\end{equation}
with the  density  vectorial function $\phi$ satisfies
\begin{equation}\label{newslt}
 {\mathcal{W}}_\Omega\phi\mid_- = ( \frac{1}{2}I+{\mathcal{K}}_\Omega) \phi \hspace{5mm} a.e.~ on~ \partial \Omega.\\
\end{equation}

Likewise, the solution to the  Stokes perturbation problem~(\ref{new}) can be represented  by the pure double-layer potential on $\partial{\Omega_\delta}$
\begin{equation}
u^i_\delta(x)= {\mathcal{W}}^i_{\Omega_\delta} \widetilde{\phi}(x)=  \int_{\partial {\Omega_\delta}} -S_{ijk}(x,y)\nu_k(y)\widetilde{\phi}_i(y)d\sigma(y),~x\in \Omega_\delta,
 \end{equation}
with the  density  vectorial function $\widetilde{\phi}$ satisfies
\begin{equation}\label{newslt2}
 {\mathcal{W}}_{\Omega_\delta}\widetilde{\phi}\mid_- = ( \frac{1}{2}I+{\mathcal{K}}_{\Omega_\delta}){\widetilde{\phi}} \hspace{5mm} a.e.~ on~ \partial {\Omega_\delta}.
\end{equation}
From the boundary condition we have $ {\mathcal{W}}_\Omega\phi\mid_-=g$ on $\partial \Omega$ and ${\mathcal{W}}_{\Omega_\delta} \widetilde{\phi}\mid_-=g_\delta$ on $\partial\Omega_\delta$.
Consider $u_{\delta}-u$ in any closed domain $\Omega^\circ  \subset \Omega \cap \Omega_{\delta}$, from the above formulae
\begin{equation}
u_\delta(x) - u(x)= {\mathcal{W}}_{\Omega_\delta} {\widetilde{\phi}} (x)-{\mathcal{W}}_\Omega \phi (x),~ x \in \Omega^\circ. 
\end{equation}

Let $\psi _\delta(x)=x+\delta\rho(x)\nu (x)$  and $\psi _\delta(x)=x+\delta\nu (x)$ be the diffeomorphisms from $\partial \Omega$ to $\partial \Omega_\delta$ in the case $d=2$ and in the case $d=3$ respectively. The following estimates hold.

\begin{lemma}
There exists a constant $C$ depending only on  $\Omega$ such that for any function $\widetilde{\phi} \in (C(\partial\Omega_\delta))^d$, we have
$$\|({\mathcal{K}}_{\Omega_\delta} \widetilde{\phi})\circ \psi_\delta-{{\mathcal {K}}_\Omega}(\widetilde{\phi}\circ\psi_\delta)\|_{({\mathcal {C}}(\partial\Omega))^d}\leq C\delta\|\widetilde{\phi}\|_{({\mathcal {C}}(\partial\Omega_\delta))^d}.$$
\end{lemma}

\begin{proof}
Fix $x\in \partial \Omega$, we then have
$$({\mathcal{K}}^i_{\Omega_\delta}\widetilde{\phi})\circ \psi_\delta(x)-{\mathcal {K}}^i_\Omega (\widetilde{\phi}\circ\psi_\delta)(x) \hspace{6cm}$$
$$=p.v. \int_{\partial{\Omega}}[-S_{ijk}(\widetilde{x},\widetilde{y})\widetilde{\nu}_k(\widetilde{y})j_\delta(y)+S_{ijk}(x,y)\nu_k(y)] \widetilde{\phi}_j\circ\psi_\delta d\sigma(y)$$
here $j_\delta$ denotes the Jacobian of $\psi_\delta$, and has the approximation  $j_\delta(y)=1+{\mathcal{O}}(\delta)$.\\

 We can express $S_{ijk}{(\widetilde{x}},\widetilde{y})=S_{ijk}{(\widetilde{x}-\widetilde{y})}$, then due to the mean value theorem  we have the following expression.\\
 In the case $d=2$, we have
\[ -S_{ijk}{(\widetilde{x},\widetilde{y})}=-S_{ijk}(x,y)- \delta(\rho(x)\nu(x)-\rho(y)\nu(y))\nabla S_{ijk}(x-y+\theta(\rho(x)\nu(x)-\rho(y)\nu(y))),\]
with $0<\theta<\delta$.\\
In the case $d=3$, we have
 \[-S_{ijk}{(\widetilde{x},\widetilde{y})}=-S_{ijk}(x,y)-\delta (\nu(x)-\nu(y))\nabla S_{ijk}(x-y+\theta(\nu(x)-\nu(y))),~0<\theta<\delta.\]
Then combining with~(\ref{nv}) and~(\ref{nv3}) we get\\
\[-S_{ijk}{(\widetilde{x},\widetilde{y})}\widetilde{\nu}_k(\widetilde{y})=-S_{ijk}(x,y)\nu_k(y)+\delta T_{ij}(x,y)+{\mathcal{O}}(\delta),\]
with 
$$T_{ij}(x,y)=-(\rho(x)\nu(x)-\rho(y)\nu(y))\nabla S_{ijk}(x-y+\theta(\rho(x)\nu(x)-\rho(y)\nu(y)) $$
$$\hspace{6cm}- S_{ijk}(x-y)\nu_k^1(y) ~{\textrm{in the case}}~ d=2;$$
$$T_{ij}(x,y)=-(\nu(x)-\nu(y))\nabla S_{ijk}(x-y+\theta(\nu(x)-\nu(y)))~ {\textrm{in the case}}~ d=3.$$
{Denote by} 
$$T^i_\Omega (\widetilde{\phi} \circ \psi_\delta)= p.v. \int_{\partial \Omega} T_{ij}(x,y)(\widetilde{\phi}_j\circ\psi_\delta)(x)d\sigma(y),$$
 we have the equality
\begin{equation}\label{ctn}
({\mathcal{K}}^i_{\Omega_\delta} \widetilde{\phi})\circ \psi_\delta(x)={\mathcal {K}}^i_\Omega (\widetilde{\phi}\circ\psi_\delta)(x)+\delta T^i_\Omega (\widetilde{\phi} \circ \psi_\delta)+{\mathcal{O}}(\delta^2).
\end{equation}
By using the remarks~\ref{rm2} and~\ref{rm3} as well as  the decay behavior of $S_{ijk}(x,y)$ as $|x-y| \rightarrow 0$ (as mentioned in Section 2.2), we obtain that
$T_{ij}(x,y)={\mathcal{O}}(|x-y|^{-d+1})$ as $|x-y| \rightarrow 0$, 
which implies that $ T_{ij}(x,y)$ is also weakly singular. So the integral $T^i_\Omega (\widetilde{\phi}\circ \psi_\delta)$ is bounded on $({\mathcal{C}}(\partial \Omega))^d$, and we have that
$$\|T_\Omega (\widetilde{\phi}\circ \psi_\delta)\|_{({\mathcal {C}}(\partial\Omega))^d}\leq C\|\widetilde{\phi}\circ \psi_\delta\|_{({\mathcal {C}}(\partial\Omega))^d},$$
this completes the proof. 
\end{proof}
\begin{lemma}
There exists a constant C depending only on $\Omega$ and $g$  such that
$$\|\widetilde{\phi} \circ \psi_\delta-\phi\|_{({\mathcal {C}}(\partial \Omega))^d}\leq \delta C(g, \Omega).$$
\end{lemma}

\begin{proof}
Since the integral operator $\frac{1}{2}I+{\mathcal{K}}_\Omega$ is invertible on $({\mathcal {C}}(\partial \Omega))^d$, combining with the trace relation~(\ref{newslt}),~(\ref{newslt2}) and Lemma~1 we obtain:
$$\|\widetilde{\phi} \circ \psi_\delta-\phi\|_{({\mathcal {C}}(\partial \Omega))^d}\leq C\|(\frac{1}{2}I+{\mathcal{K}}_\Omega)(\widetilde{\phi} \circ \psi_\delta-\phi)\|_{({\mathcal {C}}(\partial \Omega))^d} \hspace{3cm}$$
$$\hspace{1cm}\leq C\|((\frac{1}{2}I+{\mathcal{K}}_{\Omega_\delta} \widetilde{\phi})\circ\psi_\delta-(\frac{1}{2}I+{\mathcal{K}}_\Omega\phi)\|_{({\mathcal {C}}(\partial \Omega))^d}$$
$$+C\|({\mathcal{K}}_{\Omega_\delta}\widetilde{\phi})\circ\psi_\delta-{\mathcal{K}}_\Omega (\widetilde{\phi}\circ\psi_\delta)\|_{({\mathcal {C}}(\partial \Omega))^d}$$
$$\hspace{1cm}\leq C \|g_\delta\circ \psi_\delta-g\|_{({\mathcal {C}}(\partial \Omega))^d} + C\delta\|\widetilde{\phi}\|_{(C(\partial\Omega_\delta))^d}.$$
Due to the mean value theorem and the analyticity of $g$ in a neighborhood of $\partial \Omega$,  we can easily see that there exists a constant C depending only on $\Omega$ and $g$ such that $\|g_\delta\circ \psi-g\|_{({\mathcal {C}}(\partial \Omega))^d}\leq \delta C(g, \Omega) $.\\
\quad We also have that $\|\widetilde{\phi}\|_{(C(\partial\Omega_\delta))^d}\leq \|g_\delta\|_{({\mathcal {C}}(\partial \Omega))^d}$ because of the invertibility of $\frac{1}{2}+{\mathcal{K}}_{\Omega_\delta}$. \\
\quad Finally we get that $ \|\widetilde{\phi} \circ \psi_\delta-\phi\|_{({\mathcal {C}}(\partial \Omega))^d}\leq \delta C(g, \Omega).$ 
\end{proof}
\begin{theorem}
Let $\Omega^\circ$ is any closed subset of $\Omega \cap \Omega_\delta$, there exists a constant $C$ depending only  on $g$ and $\Omega$ such that $$\|u_\delta(x)-u(x)\|_{({\mathcal {C}}(\Omega^\circ))^d} \leq \delta C(g,\Omega).$$
\end{theorem}
\begin{proof}
As the double-layer potential is continuous  on  $\Omega^\circ$, by using the extreme value theorem, there is a point $x_0 \in \Omega^\circ$ such that
$$u_\delta(x_0)-u(x_0)={\mathcal{W}}_{\Omega_\delta} \widetilde{\phi} (x_0)-{\mathcal{W}}_\Omega \phi (x_0)=\max_{x\in\Omega^\circ}\|{\mathcal{W}}_{\Omega_\delta} \widetilde{\phi}(x)-{\mathcal{W}}_\Omega \phi(x)\|$$
Recall from Section~2 that we denote $-S_{ijk}(x,y)\nu_k(y)$ by ${\mathcal{D}}_{ij}(x,y)$, we then have
$$ u_\delta(x_0)-u(x_0) =\int_{\partial\Omega}[{\mathcal{D}}(x_0,\widetilde{y})j_\delta(y) \widetilde{\phi}\circ\psi_\delta(y)-{\mathcal{D}}(x_0,y)\phi(y)]d\sigma(y) $$
$$=\int_{\partial\Omega}{\mathcal{D}}(x_0, \widetilde{y})[j_\delta(y)-1]\widetilde{\phi}\circ\psi_\delta(y)d\sigma(y)+\int_{\partial\Omega}{\mathcal{D}}(x_0, \widetilde{y})[\widetilde{\phi}\circ\psi_\delta(y)-\phi(y)]d\sigma(y)$$
$$+\int_{\partial\Omega}[{\mathcal{D}}(x_0, \widetilde{y})-{\mathcal{D}}(x_0,y)]\phi(y)d\sigma(y) := I_1+I_2+I_3.$$
It follows from Lemma~2 that
$$I_2\leq \delta C_2(g,\Omega).$$
By abusing of the approximation  $j_\delta(y)=1+{\mathcal{O}}(\delta)$ we obtain $$I_1\leq \delta C_1(g, \Omega).$$
And tdue to the analyticity of ${\mathcal{D}}(x,y)$  we have that
$$I_3\leq C_3 \delta \|\phi\|_{({\mathcal {C}}(\partial \Omega))^d}\leq\delta C_4(g, \Omega).$$
The proof is completed.
\end{proof}
\section{Derivation of the asymptotic expansion}
Firstly, we  investigate the asymptotic behavior of $ {\mathcal{K}}_{\Omega_\delta} \widetilde{\phi}$ as $\delta \rightarrow 0$. Denote by $\phi_\delta={\widetilde{\phi}}\circ \psi_\delta$ the vectorial function with components $\phi_{\delta,j}~({j=\overline{1,d}})$, by the change of variable, we can rewrite the integral $ {\mathcal{K}}_{\Omega_\delta} \widetilde{\phi}$ as the following form
\begin{equation}
 {\mathcal{K}}^i_{\Omega_\delta} {\widetilde{\phi}}({\widetilde{x}})=p.v. \int_{\partial \Omega} -S_{ijk}({\widetilde{x}}, {\widetilde{y}}){\widetilde{\nu}}_k({\widetilde{y}})\phi_{\delta,j}(x)d\sigma_\delta({\widetilde{y}}).
\end{equation}
Then using Taylor  expansion  for the kernel of double-layer potential  ${\mathcal{S}}(\widetilde{x},\widetilde{y})$  at $\delta=0$ when $r=x-y \neq 0$  and combining with~(\ref{length}) and~(\ref{surface}), we obtain.\\
In the case $d=2$
\begin{equation}\label{ex}
-S_{ijk}{(\widetilde{x},\widetilde{y})}\widetilde{\nu}_k(\widetilde{y})=\underbrace{-S_{ijk}(x,y)\nu_k(y)}_{:=K_{ij}^0(x,y)}+\hspace{6cm}
\end{equation}
$$\hspace{2cm}\sum_{n=1}^\infty\delta^n \underbrace{\sum_{m+p=n} \sum_{|\alpha|=m}\frac{-((\rho(x)\nu(x)-\rho(y)\nu(y))^\alpha}{\alpha!}\nabla S^\alpha_{ijk}(x,y))\nu_k^p(y)}_{:= K_{ij}^n(x,y)},
$$
In the case $d=3$
\begin{equation}\label{ex3}
-S_{ijk}{(\widetilde{x},\widetilde{y})}\widetilde{\nu}_k(\widetilde{y})=\underbrace{-S_{ijk}(x,y)\nu_k(y)}_{:=K_{ij}^0(x,y)}+\hspace{6cm}
\end{equation}
$$\hspace{2cm}\sum_{n=1}^\infty \delta^n \underbrace{\sum_{m+p=n}\sum_{|\alpha|=m}\frac{-((\nu(x)-\nu(y))^\alpha}{\alpha!}\nabla S^\alpha_{ijk}(x,y))\nu_k^p(y)}_{:= K_{ij}^n(x,y)}.
$$

Then we introduce
 a sequence of integral operator $({\mathcal{K}}_{ij}^{n}\phi_\delta (x))_{n \in \mathbb{N}}$ defined for
any $\phi_\delta \in (C(\partial \Omega))^d$ by

\begin{equation}
{\mathcal{K}}_{\Omega}^{n,i}\phi_\delta (x)=p.v. \sum_{m+q=n}\int_{\partial \Omega}  K_{ij}^m(x,y)\sigma^q(y) \phi_{\delta,j}(y) d \sigma(y),\quad n\leq0. 
\end{equation}
Note that ${\mathcal{K}}_{\Omega}^{0}={\mathcal{K}}_{\Omega}$. By using the same arguments as in the proof for Lemma~1, we have that the operators ${\mathcal{K}}_{\Omega}^{n}{\phi_\delta}$ are bounded on $(C(\partial \Omega))^d$. \\

\begin{theorem}
Let $N\in \mathbb{N}$. There exists a constant $C$ depending on $N$ and $\Omega$  
such that for any $\widetilde{\phi} \in ({\mathcal {C}}(\partial \Omega_\delta))^d$,
$$\|({\mathcal{K}}_{\Omega_ \delta} \widetilde{\phi}) \circ \psi_\delta -{\mathcal{K}}_\Omega \phi_\delta - \sum_{n=2}^N\delta^n {\mathcal{K}}_\Omega ^{n}\phi_\delta\|_{({\mathcal {C}}(\partial \Omega))^d}\leq C \delta^{N+1}\|\phi_\delta\|_{(C(\partial \Omega))^d}$$
where $\phi_\delta:=\widetilde{\phi} \circ \psi_\delta$.
\end{theorem}
We  now investigate  the asymptotic behavior of ${\mathcal{K}}_{\Omega_\delta} \widetilde{\phi}$ as $\delta \rightarrow 0$. One can see from Theorem~2 that for each integer $N$, $\phi_\delta$ satisfies
$$(\frac{1}{2}I + {\mathcal{K}}_\Omega+\sum_{n=1}^N \delta^n {\mathcal{K}}_\Omega^{n})\phi_\delta+0(\delta ^{N+1})=g\circ \psi_\delta ~~~~~\mbox{on}~ \partial \Omega.$$
Now using  the Taylor development  for  $g(\widetilde{x})=(g_i(\widetilde{x}))_{i=\overline{1,d}}$ at $\delta=0$ we obtain the following expansions.\\
In the case $d=2$
\begin{equation}
 g_i(\widetilde{x})=\sum_{n=0}^\infty \sum_{|\alpha|=n} \frac{\delta^n}{\alpha!}\nabla^\alpha g_i(x) (\rho(x) \nu(x))^\alpha:=\sum_{n=0}^\infty \delta^n G_i^n(x)
\end{equation}
note that, the first three terms are given by
$$G_i^0(x)=g_i(x), ~G_i^1(x)=\rho(x) \nu(x) \cdot \nabla g_i(x), ~G_i^2(x)=\frac{1}{2} (\rho(x) \nu(x))^T \nabla^2 g_i(x)(\rho(x) \nu(x));$$
where the superscript $T$ denotes the transpose of a vector.\\
In the case $d=3$, we have:
\begin{equation}
g_i(\widetilde{x})=\sum_{n=0}^\infty \sum_{|\alpha|=n}  \frac{\delta^n}{\alpha!}\nabla^\alpha g_i(x) (\nu(x))^\alpha:=\sum_{n=0}^\infty \delta^n G_i^n(x)
\end{equation}
with the first three terms are given by
$$G_i^0(x)=g_i(x),~ G_i^1(x)= \nu(x) \cdot \nabla g_i(x), ~G_i^2(x)=\frac{1}{2} ( \nu(x))^T \nabla^2 g_i(x)( \nu(x)).$$

Therefore, we obtain the following integral equation to
solve
\begin{equation}\label{4}
(\frac{1}{2}I + {\mathcal{K}}_\Omega+\sum_{n=1}^N \delta^n {\mathcal{K}}^{n})\phi_\delta+0(\delta ^{N+1})=\sum_{n=0}^\infty \delta^n G^n.
\end{equation}

We look for the solution of~(\ref{4}) in the form of power series
$$\phi_\delta=\phi^0 + \sum_{n=1}^N \delta^n \phi ^{n}.$$

The equation~(\ref{4}) then can be solved recursively in the following way: Define

\begin{equation}
\phi^{0}=(\frac{1}{2} I + {\mathcal{K}})^{-1}G^0=(\frac{1}{2} I + {\mathcal{K}})^{-1}g=\phi
\end{equation}
and for $1 \leq n \leq N$,
\begin{equation}\label{5}
\phi^{n}=(\frac{1}{2} I + {\mathcal{K}})^{-1}\left (G^n+\sum^{n-1}_{p=0}{\mathcal{K}}^{n-p}\phi^{p} \right ).
\end{equation}
We obtain the following Lemma.

\begin{lemma}
Let $N\in\mathbb{N}$. There exists a constant $C$ depending only on $N$ and $\Omega$ 
such that
$$\|\phi_\delta-\sum_{n=0}^N \delta^n \phi ^{n}\|_{(C(\partial \Omega))^d} \leq C \delta^{N+1},$$
where $\phi^{n}$ are defined by the recursive relation~(\ref{5}).
\end{lemma}

Now we derive the asymptotic behavior of $u_\delta-u$ as $\delta \rightarrow 0$.
From the boundary integral  representation of the solution  in~(\ref{newslt}), after the change of variable we rewrite it as follows
\begin{equation}\label{main}
{\mathcal{W}}^i_{\Omega_\delta} \widetilde{\phi}(x)= \int _{\partial \Omega} {\mathcal{D}}_{ij}(x,\widetilde{y})\phi_{\delta,j}(y) d{\sigma}_\delta(\widetilde{y}) 
\end{equation}
where $\phi_\delta :=\widetilde{\phi} \circ \psi_\delta$, and $\{\phi_{\delta,j}\}_{j=\overline{1,d}}$ are components of $\phi_\delta$.\\

 Then due to the Taylor expansion of $S_{ijk}(x,\widetilde{y})$ with respect to $y$, combining with~(\ref{nv}) and~(\ref{nv3}), we obtain\\
 
 For the case $d=2$ 
\begin{equation}\label{D}
- S_{ijk}(x,\widetilde{y})\widetilde{\nu}_k(\widetilde{y})=- S_{ijk}(x,y)\nu_k(y)\hspace{6cm}
\end{equation}
$$\hspace{2cm}+\sum_{n=1}^\infty\delta^n \underbrace{ \sum_{m+p=n}\sum_{|\alpha|=m}\frac{-(\rho(y)\nu(y))^\alpha}{\alpha!} \nabla^\alpha_y S_{ijk}(x,y)\nu^p_k(y)}_{:= {\mathcal{D}}^n_{ij}(x,y)}.$$
Likewise, in the case $d=3$ 
\begin{equation}\label{D}
- S_{ijk}(x,\widetilde{y})\widetilde{\nu}_k(\widetilde{y})=- S_{ijk}(x,y)\nu_k(y)+ \sum_{n=1}^\infty\delta^n\underbrace{ \sum_{|\alpha|=n}\frac{-(\nu(y))^\alpha}{\alpha!} \nabla^\alpha_y S_{ijk}(x,y)\nu_k(y)}_{:= {\mathcal{D}}^n_{ij}(x,y)}.
 \end{equation}

Thanks to ~(\ref{length}),~(\ref{surface}) and ~(\ref{D}), the formula~(\ref{main}) can be rewritten as follows
\begin{equation}
{\mathcal{W}}^i_{{\Omega_\delta}}[ \widetilde{\phi}](x)= \int _{\partial \Omega}\left ( {\mathcal{D}}_{ij}(x,y)+\sum_{n=1}^\infty \delta^n {\mathcal{D}}^n_{ij}(x,y)\right )\hspace{3cm}
\end{equation}
$$\hspace{2cm}\left( \phi_j(y)+ \sum_{n=1}^N  \delta^n \phi_j^{n}(y)+{\mathcal{O}}(\delta^{N+1})\right) \left( 1+\sum_{n=1}^\infty\delta ^n \sigma^n(y)\right )d\sigma(y) $$
$$=\int _{\partial \Omega}{\mathcal{D}}_{ij}(x,y)\phi_j(y)d\sigma(y)+\sum_{n=1}^N \delta^n\sum_{m+k+q=n} \int _{\partial \Omega}   {\mathcal {D}}^m_{ij} (x,y) \phi_j^{k} (y) \sigma ^{q}(y)  d\sigma(y)+{\mathcal{O}}(\delta^{N+1}).$$
Note that $\int _{\partial \Omega}{\mathcal{D}}_{ij}(x,y)\phi_j(y)d\sigma(y)={\mathcal{W}}^i_{{\Omega}} {\phi}(x)$. We then  define, for $n\in \mathbb{N}$ and for $x \in  \Omega^\circ$, the terms $u_n=(u^i_{n})_{i=\overline{1,d}}$ as follows
\begin{equation}\label{v}
u^i_{n}(x)=\sum_{m+q+k=n} \int _{\partial \Omega}   {\mathcal {D}}^m_{ij} (x,y) \sigma ^{q}(y) \phi_j^{k} (y) d\sigma(y).
\end{equation}
We obtain that
\begin{equation}\label{im}
u_\delta - u= \sum_{n=1}^N \delta^n u_n(x)+{\mathcal{O}}(\delta^{N+1}),~ x\in \Omega^\circ.
\end{equation}
The remainder $O(\delta^{N+1})$ depends only on $N$ and $\Omega$. \\

Let us compute the first order approximation of $u_\delta-u$ explicitly. Note that $\phi^{0}=\phi$ where $\phi$ is defined by~(\ref{5}) and
\begin{equation}\label{phiex}
\phi^{1}=(\frac{1}{2}I+{\mathcal{K}}_\Omega)^{-1}(G^1+{\mathcal{K}}_\Omega^1\phi).
\end{equation}

Therefore, in the two-dimensional case,  $u_1=(u^i_1)_{i=\overline{1,2}}$ takes the form
\begin{equation}\label{ex1}
u^i_{1}(x)=\int_{\partial \Omega} {\mathcal{D}}^1_{ij}(x,y) \phi_j(y)d\sigma(y)-\int_{\partial \Omega}  {\mathcal{D}}_{ij}(x,y)\tau(y)\rho(y)\phi_j(y)d\sigma(y) 
\end{equation}
$$\hspace{8cm}+\int_{\partial \Omega} {\mathcal{D}}_{ij}(x,y)\phi_j^{1}(y)d\sigma(y).$$
And in the three-dimensional case,  $u_1=(u^i_1)_{i=\overline{1,3}}$ takes the form
\begin{equation}\label{ex2}
u^i_{1}(x)=\int_{\partial \Omega}  {\mathcal{D}}^1_{ij}(x,y) \phi_j(y)d\sigma(y)-2\int_{\partial \Omega} { \mathcal{D}}_{ij}(x,y)H(x)\phi_j(y)d\sigma(y) 
\end{equation}
$$\hspace{8cm}+\int_{\partial \Omega} {\mathcal{D}}_{ij}(x,y)\phi_j^{1}(y)d\sigma(y).$$

\end{document}